\documentclass[twoside,leqno,twocolumn]{article}

\usepackage{ltexpprt}

\usepackage{cite}
\usepackage{amsmath,amssymb,amsfonts}
\usepackage{enumitem}
\usepackage{algorithm}
\usepackage[noend]{algpseudocode}
\usepackage{graphicx}
\usepackage{textcomp}
\usepackage{xcolor}
\usepackage{hyperref}
\usepackage{cleveref}
\usepackage[font=small,labelfont=bf]{caption}
\usepackage{comment}
\newcommand{\isep}{\mathrel{{.}\,{.}}\nobreak}

\def\BibTeX{{\rm B\kern-.05em{\sc i\kern-.025em b}\kern-.08em
    T\kern-.1667em\lower.7ex\hbox{E}\kern-.125emX}}

\definecolor{mygreen}{HTML}{1B5E20}
\definecolor{myred}{HTML}{d62728}

\newcommand{\Name}[0]{FastSV} 

\newcommand{\todo}[1]{{\small\color{red}[#1]}}
\algrenewcommand{\algorithmiccomment}[1]{{$\triangleright$ \color{blue}{#1}}}

\crefname{lemma}{Lemma}{Lemmas}
\crefname{theorem}{Theorem}{Theorems}

\begin{document}

\title{\LARGE \Name{}: A Distributed-Memory Connected Component Algorithm with Fast Convergence}
\author{anonymous author(s)}
\author{Yongzhe Zhang\thanks{SOKENDAI, Japan}
\and Ariful Azad\thanks{Indiana University Bloomington, USA}
\and Zhenjiang Hu\thanks{Peking University, China}}

\date{}

\maketitle



\begin{abstract}
This paper presents a new distributed-memory algorithm called \Name{} for finding connected components in an undirected graph.
Our algorithm simplifies the classic Shiloach-Vishkin algorithm and employs several novel and efficient hooking strategies for faster convergence.
We map different steps of \Name{} to linear algebraic operations and implement them with the help of scalable graph libraries.
\Name{} uses sparse operations to avoid redundant work and optimized MPI communication to avoid bottlenecks. 
The resultant algorithm shows high-performance and scalability as it can find the connected components of a hyperlink graph with over 134B edges in $30$ seconds using 262K cores on a Cray XC40 supercomputer.
\Name{}  outperforms the state-of-the-art algorithm by an average speedup of $2.21\times$ (max $4.27\times$) on a variety of real-world graphs.

\end{abstract}

\section{Introduction}
\label{sec:introduction}


This paper presents a distributed-memory parallel algorithm for finding connected components (CC) in an undirected graph $G(V, E)$ where $V$ and $E$ are the set of vertices and edges, respectively.
A connected component is a subgraph of $G$ in which every pair of vertices are connected by paths and no vertex in the subgraph is connected to any other vertex outside of the subgraph. 
Finding connected components has numerous applications in bioinformatics~\cite{van2000graph}, computer vision~\cite{yang1989improved}, and scientific computing~\cite{pothen1990computing}.

Sequentially, connected components of a graph with $n$ vertices and $m$ edges can be easily found by breadth-first search (BFS) or depth-first search in $O(m+n)$ time.
While this approach performs linear work, the depth is proportional to the sum of the diameters of the connected components.
Therefore, BFS-based parallel algorithms 
are not suitable for high-diameter graphs or graphs with millions of connected components.
Connectivity algorithms based on the ``tree hooking" scheme work by arranging the vertices into disjoint trees such that at the end of the algorithm, all vertices in a tree represent a connected component.
Shiloach and Vishkin~\cite{sv} used this idea to develop a highly-parallel PRAM (parallel random access machine) algorithm that runs in $O(\log n)$ time using $O(n+m)$ processors.
Their algorithm is not work efficient as it performs $O(m\log n)$ work, but the availability of $O(m)$ parallel work made it an attractive choice for large-scale distributed-memory systems.  
Therefore, the  Shiloach-Vishkin (SV) algorithm and its variants are frequently used in scalable distributed-memory CC algorithms such as LACC~\cite{lacc}, ParConnect~\cite{parconnect}, and Hash-Min~\cite{hashmin}.

To the best of our knowledge, LACC~\cite{lacc} is the most scalable published CC algorithm that scales to 262K cores when clustering graphs with more than 50B edges.
LACC is based on the Awerbuch-Shiloach (AS) algorithm, which is a simplification of the SV algorithm.
The AS algorithm consists of four steps: (a) finding stars (trees of height 1), (b) hooking stars conditionally onto other trees, (c) hooking stars unconditionally onto other trees, (d) shortcutting to reduce the height of trees. 
LACC mapped these operations to parallel linear-algebraic operations such as those defined in the GraphBLAS standard~\cite{graphblas} and implemented them in the CombBLAS~\cite{combblas} library for scalability and performance.
We observed that  LACC's requirements of star hooking  and unconditional hooking can be safely removed to design a simplified SV algorithm with just two steps: (a) hooking trees conditionally onto other trees and (b) shortcutting. 
After mapping these two operation to linear algebra and performing other simplifications, we developed a distributed-memory SV algorithm that is both simpler and faster than LACC. 
Since, each of the four operations in LACC takes about 25\% of the total runtime, each iteration of our SV is usually more than $2\times$ faster than each iteration of LACC when run on the same number of processors. 
However, the simplified SV requires more iterations than LACC because of the removal of unconditional hooking. 
To alleviate this problem, we developed several novel hooking strategies for faster convergence, hence the new algorithm is named as \Name{}.

The simplicity of \Name{} along with its fast convergence schemes makes it suitable for distributed-memory platforms. 
We map different steps of \Name{}  to linear-algebraic operations and implemented the algorithm using the CombBLAS library.
We choose CombBLAS due to its high scalability and the fact that the state-of-the-art connected component algorithm LACC and ParConnect also rely on CombBLAS.
We further employ several optimization techniques for eliminating  communication bottlenecks.
The resultant algorithm is highly parallel as it scales up to $262,144$ cores of a Cray XC40 supercomputer and can find CCs from graphs with billions of vertices and hundreds of billions of edges in just 30 seconds.
\Name{} advances the state-of-the-art in parallel CC algorithm as it is on average $2.21\times$ faster than the previous fastest algorithm LACC.
Overall, we made the following technical contributions in this paper:
\begin{itemize}[leftmargin=*]
\item We develop a simple and efficient algorithm \Name{} for finding connected components in distributed memory. \Name{} uses novel hooking strategies for fast convergence. 
\item We present \Name{} using a handful of GraphBLAS operations and implement the algorithm in CombBLAS for distributed-memory platforms and in SuiteSparse:GraphBLAS for shared-memory platforms. We dynamically use sparse operations to avoid redundant work and optimize MPI communication to avoid bottlenecks. 

\item Both shared- and distributed-memory implementations of \Name{} are significantly faster than the state-of-the-art algorithm LACC. The distributed-memory implementation of \Name{} can find CCs in a hyperlink graph with 3.27B vertices and 124.9B edges in just $30$ seconds using 262K cores of a XC40 supercomputer. 
\end{itemize}

\section{Background}

\subsection{Notations.}

This paper assumes the connected component algorithm to be performed on an undirected graph $G=(V,E)$ with $n$ vertices and $m$ edges. 
For each vertex $v\in V$, we use $N(v)$ to denote $v$'s neighbors, the set of vertices adjacent to $v$.
We use \emph{pointer graph} to refer to an auxiliary directed graph $G_p=(V,E_p)$ for $G$, where for every vertex $v\in V$ there is exactly one directed edge $(v,v_1)\in E_p$ and $v_1\leq v$.
If we ignore the self-loops $(v,v)\in E_p$, $G_p$ defines a forest of directed rooted trees where every vertex can follow the directed edges to reach the root vertex.
In $G_p$, a tree is called a \emph{star} if every vertex in the tree points to a root vertex (a root points to itself).

\subsection{GraphBLAS.}

Expressing graph algorithms in the language of linear algebra is appealing.
By using a small set of matrix and vector (linear algebra) operations, many scalable graph algorithms can be described clearly, the parallelism is hidden for the programmers, and the high performance can be achieved by performance experts implementing those primitives on various architectures.
Several independent systems have emerged that use matrix algebra to perform graph computations~\cite{combblas,gpi,graphmat}.
Recently, GraphBLAS~\cite{graphblas} defines a standard set of linear-algebraic operations (and C APIs~\cite{graphblas-C}) for implementing graph algorithms.
In this paper, we will use the functions from the GraphBLAS API to describe our algorithms due to its conciseness.
Our distributed implementation is based on CombBLAS~\cite{combblas}.

\subsection{The original SV algorithm.}
The SV algorithm stores the connectivity information in a forest of rooted trees, where each vertex $v$ maintains a field $f[v]$ through which it points to either itself or another vertex in the same connected component.
All vertices in a tree belong to the same component, and at termination of the algorithm, all vertices in a connected component belong to the same tree. 
Each tree has a designated root (a vertex having a self-loop) that serves as the representative vertex for the corresponding component.
This data structure is called a pointer graph, which changes dynamically during the course of the algorithm.

The algorithm begins with $n$ single-vertex trees and iteratively merges trees to find connected components. 
Each iteration of the original SV algorithm performs a sequence of four operations: (a) conditional hooking,  (b) shortcutting, (c) unconditional hooking and (d) another shortcutting.
Here, hooking is a process where the root of a tree becomes a child of a vertex from another tree.
Conditional hooking of a root $u$ is allowed only when $u$'s id is larger than the vertex which $u$ is hooked into. 
Unconditional hooking can hook any trees that remained unchanged in the preceding conditional hooking.
The shortcutting step reduces the height of trees by replacing a vertex's parent by its grandparent.
With these four steps the SV algorithm is guaranteed to finish in $O(\log n)$ iterations, where each iteration performs $O(m)$ parallel work. 

The original Shiloach-Vishkin paper mentioned that the last shortcutting is for a simpler proof of their algorithm. 
Hence, it can be removed without sacrificing correctness or convergence speed.
If we remove unconditional hooking, the algorithm is still correct, but it may need $O(n)$ iterations in the worst case.
Nevertheless, practical parallel algorithms often remove the unconditional hooking~\cite{palgol,cc-ipdps18} because it needs to keep track of unchanged trees (also known as stagnant trees), which is expensive, especially in distributed memory.
We follow the same route and use a simplified SV algorithm discussed next. 


\begin{algorithm}[t]
\small
\caption{The SV algorithm. \textbf{Input:} An undirected graph $G(V,E)$. \textbf{Output:} The parent vector $f$.}
\label{algo:algo-one}
\begin{algorithmic}[1]
\Procedure {SV}{$V,E$}
\For {every vertex $u \in V$}
	\State $f[u], f_{\mathit{next}}[u] \leftarrow u$
\EndFor
\Repeat
	\State\Comment {Step 1: Tree hooking}
	\For {every $(u,v) \in E$} \textbf{in parallel}
		\If {$f[u]=f[f[u]]$ \textbf{and} $f[v] < f[u]$}
			\State $f_{\mathit{next}}[f[u]] \leftarrow f[v]$
		\EndIf
	\EndFor
	\State $f\leftarrow f_{\mathit{next}}$
	\State\Comment {Step 2: Shortcutting}
	\For {every $u \in V$} \textbf{in parallel}
		\If {$f[u]\neq f[f[u]]$}
			\State $f_{\mathit{next}}[u] \leftarrow f[f[u]]$
		\EndIf
	\EndFor
	\State $f\leftarrow f_{\mathit{next}}$
\Until{$f$ remains unchanged}
\EndProcedure
\end{algorithmic}
\end{algorithm}

\subsection{A simplified SV algorithm }
\label{sec:algo-one}
\autoref{algo:algo-one} describes the simplified SV algorithm, which is the basis of our parallel algorithm. 
Initially, the parent $f[u]$ of a vertex $u$ is set to itself to denote $n$ single-vertex trees. 
We additionally maintain a copy $f_{\mathit{next}}$  of the parent vector so that the parallel algorithm reads from $f$ and writes to $f_{\mathit{next}}$.
Given a fixed ordering of vertices, each execution of \autoref{algo:algo-one} generates exactly the same pointer graph after the $i$th iteration because of using separate vectors for reading and writing. 
Hence, the convergence pattern of this parallel algorithm is completely deterministic, making it suitable for massively-parallel distributed systems.  
By contrast, concurrent reading from and writing to a single vector $f$ still deliver the correct connected components, but the structures of intermediate pointer graphs are not deterministic.

In each iteration, the algorithm performs tree hooking and shortcutting operations in order:
\begin{itemize}
\item \textbf{Tree hooking (line 6--8):} for every edge $(u,v)$, if $u$'s parent $f[u]$ is a root and $f[v]<f[u]$, then make $f[u]$ point to $f[v]$. As mentioned before, the updated parents are stored in a separate vector $f_{\mathit{next}}$ so the updated parents are not used in the current iteration. 
\item \textbf{Shortcutting (line 11--13):} if a vertex $u$ does not point to a root vertex, make $u$ point to its grandparent $f[f[u]]$.
\end{itemize}

The algorithm terminates when the parent vector remains unchanged in the latest iteration.
At termination, every tree becomes a star, and vertices in a star constitute a connected component. 
The correctness of this algorithm is discussed in previous work~\cite{greiner1994comparison}.
However, as mentioned before, without the unconditional hooking used in the original SV algorithm, we can no longer guarantee that \autoref{algo:algo-one} converges in $O(\log n)$ iterations. 
We will show in Section~\ref{sec:evaluation} that \autoref{algo:algo-one} indeed converges slowly, but does not require the worst case bound $O(n)$ iterations for the practical graphs we considered. 
Nevertheless, the extra iterations needed by \autoref{algo:algo-one} increase the runtime of parallel SV algorithms. 
To alleviate this problem, we develop several novel hooking schemes, ensuring that the improved algorithm \Name{}  is as simple as \autoref{algo:algo-one}, but the former converges faster than the latter.

\section{The \Name{} algorithm}
\label{sec:algo-two}

In this section, we introduce four important optimizations for the simplified SV algorithm, obtaining \Name{} with faster convergence.

\subsection{Hooking to grandparent.}

\begin{figure}[t]
\centering
\includegraphics[width=0.27\textwidth]{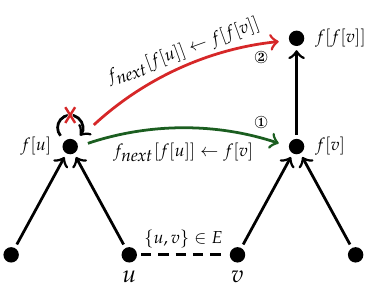}
\vspace{-10pt}
\caption{Two different ways of performing the tree hooking. (1) the original algorithm that hooks $u$'s parent $f[u]$ to $v$'s parent $f[v]$, (2) hook $u$'s parent $f[u]$ to $v$'s grandparent $f[f[v]]$. Both strategies are correct and the latter one improves the convergence.}
\vspace{-10pt}
\label{fig:hooking}
\end{figure}

In the original algorithm, the tree hooking is represented by the assignment $f_{\mathit{next}}[f[u]]\leftarrow f[v]$ (line 8 in \autoref{algo:algo-one}) requiring $f[u]$ to be a root vertex, $(u,v)\in E$ and $f[v]<f[u]$.
It is not hard to see, if we perform the tree hooking using $v$'s grandparent $f[f[v]]$, saying $f_{\mathit{next}}[f[u]]\leftarrow f[f[v]]$, the algorithm will still produce the correct answer.
To show this, we visualize both operations in \autoref{fig:hooking}.

Suppose $(u,v)$ is an edge in the input graph and $f[v]<f[u]$.
The original hooking operation is represented by the green arrow in the figure, which hooks $f[u]$ to $v$'s parent $f[v]$.
Then, our new strategy simply changes $f[v]$ to $v$'s grandparent $f[f[v]]$, resulting the red arrow from $f[u]$ to $f[f[v]]$.
It is not hard to see, as long as we choose a value like $f[f[v]]$ such that it is in the same tree of $v$, we can easily prove the correctness of the algorithm.
One can also expect that any value like $f^{k}[v]$ ($v$'s $k$-th level ancestor) will also work.

Intuitively, choosing a higher ancestor of $v$ in the tree hooking will likely create shorter trees, leading to faster convergence (all trees are stars at termination).
However, finding higher ancestors may incur additional computational cost.  
Here, we choose grandparents $f[f[v]]$ because they are needed in the shortcutting operation anyway; hence, using grandparents does not incur additional cost in the hooking operation.

\begin{figure}
\centering
\includegraphics[width=0.48\textwidth]{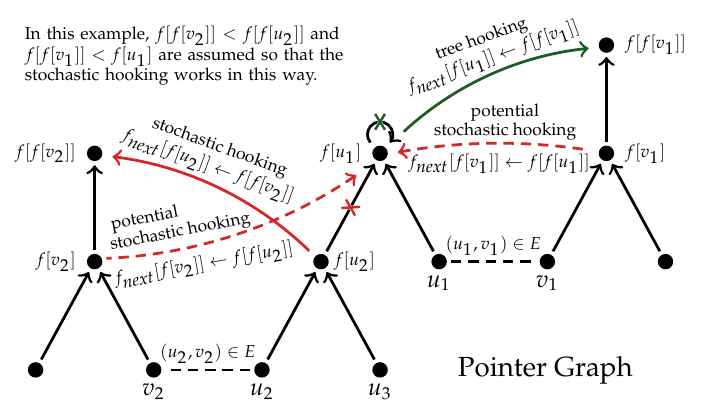}
\caption{The stochastic hooking strategy.
Suppose there are two edges $(u_1,v_1)$ and $(u_2,v_2)$ activating the hooking operation.
The red arrows are the potential modifications to the pointer graph due to our stochastic hooking strategy, which tries to hook a non-root vertex to another vertex.
The solid line successfully modifies $f[u_2]$'s pointer to $f[f[v_2]]$, but the dashed lines do not take effect due to the ordering on the vertices.}
\label{fig:stochastic}
\vspace{-10pt}
\end{figure}

\subsection{Stochastic hooking.}
The original SV algorithm and \autoref{algo:algo-one} always hooked the root of a tree onto another tree (see \autoref{fig:hooking} for an example).
Therefore, the hooking operation in \autoref{algo:algo-one} never breaks a tree into multiple parts and hooks different parts to different trees.
This restriction is enforced by the equality check $f[f[u]]=f[u]$ in line 7 of \autoref{algo:algo-one}, which is only satisfied by roots and their children. 
We observed that this restriction is not necessary for the correctness of the SV algorithm.
Intuitively, we can split a tree into multiple parts and hook them independently because these tree fragments will eventually be merged to a single connected component when the algorithm terminates. 
We call this strategy \emph{stochastic hooking}.



The stochastic hooking strategy can be employed by simply removing the condition $f[f[u]]=f[u]$ from line 7 of \autoref{algo:algo-one}.
Then, any part of a tree is allowed to hook onto another vertex when the other hooking conditions are satisfied. 
It should be noted that after removing the condition $f[f[u]]=f[u]$, it is possible that a tree may hook onto a vertex in the same tree.
This will not affect the correctness though.
In this case, the effect of stochastic hooking is similar to the shortcutting, which hooks a vertex to some other vertex with a smaller identifier.

\autoref{fig:stochastic} shows an example of stochastic hooking by the solid red arrow from $f[u_2]$ to $f[f[v_2]]$.
In the original algorithm, $u_2$ does not modify its non-root parent $f[u_2]$'s pointer, but stochastic hooking changes $f[u_2]$'s pointer to one of $u$'s neighbor's grandparent $f[f[v_2]]$.
Suppose $f[u_1]$ points to $f[f[v_1]]$ after the tree hooking, we can see that $f[u_1]$ and $f[u_2]$ might be no longer in the same connected component (assuming $f[f[v_1]]$ and $f[f[v_2]]$ are in different trees).
Possible splitting of trees is a policy that differs from the conventional SV algorithm, but it gives a non-root vertex an opportunity to be hooked.
In \autoref{fig:stochastic}, $f[u_2]$'s new parent $f[f[v_2]]$ is smaller than $f[u_1]$, which can expedite the convergence.

\autoref{algo:algo-two} presents the high-level description of \Name{} using the new hooking strategies.
Here, $\xleftarrow{\min}$ denotes a compare-and-assign operation that updates an entry of $f_{\mathit{next}}$ only when the right hand side is smaller.
The stochastic hooking is shown in line 6--7, and the shortcutting operation in line 12--13 is also affected by the removal of the predicate $f[f[u]]=f[u]$.

\begin{figure}
\centering
\includegraphics[width=0.48\textwidth]{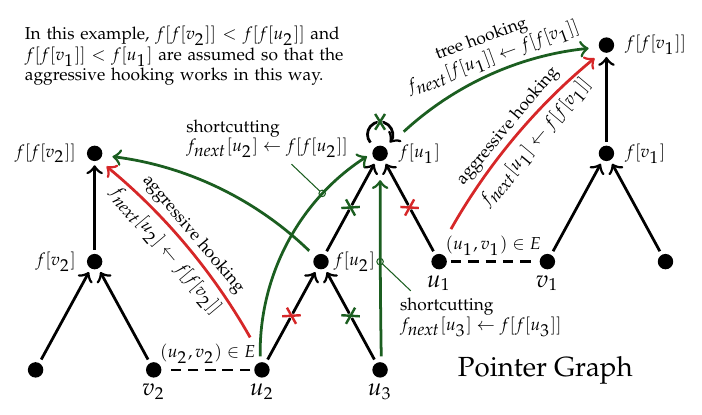}
\caption{The aggressive hooking strategy.
Suppose there are two edges $(u_1,v_1)$ and $(u_2,v_2)$ activating the hooking operation.
The green arrows represent the hooking strategies introduced so far, and the red arrows represent our aggressive hooking strategy where a vertex may point at one of its neighbor's grandparent.
Some vertices may have multiple arrows (like $u_2$), and which vertex to hook onto is decided by the ordering on the vertices.}
\label{fig:aggressive}
\vspace{-10pt}
\end{figure}

\subsection{Aggressive hooking.}
Next, we give a vertex $u$ another chance to hook itself onto another tree whenever possible.
This strategy is called \emph{aggressive hooking}, performed by $f_{\mathit{next}}[u]\xleftarrow{\min} f[f[v]]$ for all $(u,v)\in E$.
\autoref{fig:aggressive} gives an example of aggressive hooking by the red arrow for $u_1$ and $u_2$.
Here, $u_1$'s pointer will not be modified by any hooking operation introduced so far.
Then, the aggressive hooking makes $u_1$ point to its newest grandparents $f[f[v_1]]$, as if an additional shortcutting is performed.
We should mention that the cost of an additional shortcutting $f'_{\mathit{next}}\leftarrow f_{\mathit{next}}[f_{\mathit{next}}]$ is expensive due to the recalculation of grandparents, while the aggressive hooking is essentially a cheap element-wise operation over $f$ by reusing some results in the stochastic hooking.
We will discuss how they are implemented \autoref{sec:sv-graphblas}.

For $u_2$ in \autoref{fig:aggressive}, it only performs the shortcutting operation $f_{\mathit{next}}[u_2]\leftarrow f[f[u_2]]$ in the original algorithm, and now the aggressive hooking performs $f_{\mathit{next}}[u_2]\leftarrow f[f[v_2]]$.
Our implementation let $f[u_2]$ point to the smaller one between $f[f[u_2]]$ and $f[f[v_2]]$, which is expected to give the best convergence for vector $f$.

\begin{algorithm}[t]
\small
\caption{The \Name{} algorithm. \textbf{Input:} $G(V,E)$. \textbf{Output:} The parent vector $f$}
\label{algo:algo-two}
\begin{algorithmic}[1]
\Procedure {\Name{}}{$V,E$}
\For {every vertex $u \in V$}
	\State $f[u], f_{\mathit{next}}[u] \leftarrow u$
\EndFor
\Repeat
	\State\Comment {Step 1: Stochastic hooking}
	\For {every $(u,v) \in E$} \textbf{in parallel}
		\State $f_\mathit{next}[f[u]] \xleftarrow{\min} f[f[v]]$
	\EndFor
	\State\Comment {Step 2: Aggressive hooking}
	\For {every $(u,v) \in E$} \textbf{in parallel}
		\State $f_{\mathit{next}}[u] \xleftarrow{\min} f[f[v]]$
	\EndFor
	\State\Comment {Step 3: Shortcutting}
	\For {every $u \in V$} \textbf{in parallel}
		\State $f_{\mathit{next}}[u] \xleftarrow{\min} f[f[u]]$
	\EndFor
	\State $f\leftarrow f_{\mathit{next}}$
\Until{$f[f]$ remains unchanged}
\EndProcedure
\end{algorithmic}
\end{algorithm}

\subsection{Early termination.}
\label{sec:termination}

The last optimization is a generic one that applies to most variations of the SV algorithm.
SV's termination is based on the stabilization of the parent vector $f$, which means even if $f$ reaches the converged state (where every vertex points to the smallest vertex in its connected component), we need an additional iteration to verify that.
We will see in \autoref{sec:convergence} that for most real-world graphs, \Name{} usually takes 5 to 10 iterations to converge.
Hence, this additional iteration can consume a significant portion of the runtime. 
The removal of the last iteration is possible by detecting the stabilization of the grandparent $f[f]$ instead of $f$.
The following lemma ensures the correctness of this new termination condition.
\begin{lemma}\label{thm1}
After an iteration, if the grandparent $f[f]$ remains unchanged, then the vector $f$ will not be changed afterwards.
\end{lemma}

\textit{Proof.} See \autoref{sec:thm1}.

In practice, we found that on most practical graphs, \Name{} identifies all the connected components before converged, and the last iteration always performs the shortcutting operation to turn the trees into stars.
In such case, the grandparent vector $f[f]$ converges one iteration earlier than $f$.

\section{Implementation of \Name{} in linear algebra}

In this section, we first give the formal description of \Name{} in GraphBLAS, a standardized set of linear algebra primitives for describing graph algorithms.
We then present its linear algebra distributed-memory implementation in Combinatorial BLAS~\cite{combblas} and discuss two optimization techniques for improving its performance.


\subsection{Implementation in GraphBLAS.}
\label{sec:sv-graphblas}

In GraphBLAS, we assume that the vertices are indexed from $0$ to $|V|-1$, then the vertices and their associated values are stored as GraphBLAS object \texttt{GrB\_Vector}.
The graph's adjacency matrix is stored as a GraphBLAS object \texttt{GrB\_Matrix}.
For completeness, we concisely describe the GraphBLAS functions used in our implementation below, where the formal descriptions of these functions can be found in the API document~\cite{graphblas}.
We use $\emptyset$ to denote \texttt{GrB\_NULL}, which is fed to those ignored input parameters.
\begin{itemize}[leftmargin=*]
\item
  The function $\texttt{GrB\_mxv}(\mathit{y}, \emptyset, \text{accum}, \text{semiring}, \mathbf{A}, \mathit{x}, \emptyset)$ multiplies the matrix $\mathbf{A}$ with the vector $\mathit{x}$ on a semiring and outputs the result to the vector $\mathit{y}$.
  When the accumulator (a binary operation \textit{accum}) is specified, the multiplication result is combined with $y$'s original value instead of overwriting it.
\item
  The function $\texttt{GrB\_extract}(\mathit{y}, \emptyset, \emptyset, \mathit{x}, \text{index}, \text{n}, \emptyset)$ extracts a sub-vector $y$ from the specified positions in an input vector $x$.
  We can regard this operation as $y[i]\leftarrow x[\text{index}[i]]$ for $i\in[0\isep n-1]$ where $n$ is the length of the array $\mathit{index}$ and also the vector $\mathit{y}$.
\item
  The function $\texttt{GrB\_assign}(\mathit{y}, \emptyset, \text{accum}, \mathit{x}, \text{index}, \text{n}, \emptyset)$ assigns the entries from the input vector $x$ to the specified positions of an output vector $y$.
  We can regard it as $y[\text{index}[i]]\leftarrow x[i]$ for $i\in[0\isep n-1]$ where $n$ is the length of the array $\mathit{index}$ and also the vector $x$.
  $\mathit{accum}$ is the same as the one in \texttt{GrB\_mxv}.
\item
  The function $\texttt{GrB\_eWiseMult}(\mathit{y}, \emptyset, \emptyset, \text{binop}, \mathit{x}_1, \mathit{x}_2, \emptyset))$ performs the element-wise (generalized) multiplication on the intersection of elements of two vectors $x_1$ and $x_2$ and outputs the vector $y$.
\item
  The function $\texttt{GrB\_Vector\_extractTuples}(\text{index}, \text{value},$ $\&\text{n}, f)$ extracts the nonzero elements (tuples of index and value) from vector $f$ into two separate arrays $\mathit{index}$ and $\mathit{value}$.
  It returns the element count to $n$.
\end{itemize}
For the rest functions, we have \texttt{GrB\_Vector\_dup} to duplicate a vector, \texttt{GrB\_reduce} to reduce a vector to a scalar value through a user-specified binary operation, and \texttt{GrB\_Matrix\_nrows} to obtain the dimension of a matrix.

\autoref{algo:graphblas2} describes the  \Name{} algorithm in GraphBLAS.
Before every iteration, we calculate the initial grandparent $\mathit{gf}$ for every vertex.
First, we perform the stochastic hooking in line 10--11.
GraphBLAS has no primitive that directly implements the parallel-for on an edge list (line 9 in \autoref{algo:algo-two}), so we have to first aggregate $v$'s grandparent $\mathit{gf}[v]$ to vertex $u$ for every $(u,v)\in E$, obtaining the vector $\mathit{mngf}[u]=\min_{v\in N(u)}\mathit{gf}[v]$.
This can be implemented by a matrix-vector multiplication $\mathit{mngf}=\textbf{A}\cdot\mathit{gf}$ using the (select2nd, min) semiring.
Next, the hooking operation is implemented by the assignment $f[f[u]]\leftarrow \mathit{mngf}[u]$ for every vertex $u$.
This is exactly the \texttt{GrB\_assign} function in line 10 where the indices are the values of vector $f$ extracted in either line 5 before the first iteration or line 16 from the previous iteration.
The accumulator \texttt{GrB\_MIN} prevents the nondeterminism caused by the modification to the same entry of $f$, and the minimum operation gives the best convergence in practice.

Aggressive hooking is then implemented by an element-wise multiplication $f\leftarrow \min(f,\mathit{mngf})$ in line 13.
Although it is another operation in \Name{} that performs the parallel-for on an edge list, it can reuse the vector $\mathit{mngf}$ computed in the previous step, so the aggressive hooking is actually efficient.
Shortcutting is also implemented by an the element-wise multiplication $f\leftarrow \text{min}(f,\mathit{gf})$ in line 15.
Next, we calculate the grandparent vector $\mathit{gf}[u]\leftarrow{f[f[u]]}$.
It is implemented by the \texttt{GrB\_extract} function in line 18 where the indices are the values of $f$ extracted in line 17.


At the end of each iteration, we calculate the number of modified entries in $\mathit{gf}$ in line 20 -- 21 to check whether the algorithm has converged or not.
A copy of $\mathit{gf}$ is stored in the vector $\mathit{dup}$ for determining the termination in the next iteration.

\begin{algorithm}[t]
\small
\caption{The linear algebra \Name{} algorithm. \textbf{Input:} The adjacency matrix $\mathbf{A}$ and the parent vector $f$. \textbf{Output:} The parent vector $f$.}
\label{algo:graphblas2}
\begin{algorithmic}[1]
\Procedure {\Name{}}{$\mathbf{A},f$}
\State \text{GrB\_Matrix\_nrows} $(\&\text{n}, \mathbf{A})$
\State \text{GrB\_Vector\_dup} $(\&\mathit{gf}, \mathit{f})$ \Comment {initial grandparent}
\State \text{GrB\_Vector\_dup} $(\&\mathit{dup}, \mathit{gf})$ \Comment {duplication of $\mathit{gf}$}
\State \text{GrB\_Vector\_dup} $(\&\mathit{mngf}, \mathit{gf})$
\State \text{GrB\_Vector\_extractTuples} $(\text{index}, \text{value}, \&\text{n}, f)$
\State {$\text{Sel2ndMin}\leftarrow \text{a (select2nd, Min) semiring}$}
\Repeat
	\State \Comment {Step 1: Stochastic hooking}
	\State \text{GrB\_mxv} $(\mathit{mngf}, \emptyset, \text{GrB\_MIN}, \text{Sel2ndMin}, \mathbf{A}, \mathit{gf}, \emptyset)$
	\State \text{GrB\_assign} $(f, \emptyset, \text{GrB\_MIN}, \mathit{mngf}, \text{value}, \text{n}, \emptyset)$
	\State \Comment {Step 2: Aggressive hooking}
	\State \text{GrB\_eWiseMult} $(f, \emptyset, \emptyset, \text{GrB\_MIN}, f, \mathit{mngf}, \emptyset)$
	\State \Comment {Step 3: Shortcutting}
	\State \text{GrB\_eWiseMult} $(f, \emptyset, \emptyset, \text{GrB\_MIN}, f, \mathit{gf}, \emptyset)$
	\State \Comment {Step 4: Calculate grandparents}
	\State \text{GrB\_Vector\_extractTuples} $(\text{index}, \text{value}, \&\text{n}, f)$
	\State \text{GrB\_extract} $(\mathit{gf}, \emptyset, \emptyset, f, \text{value}, \text{n}, \emptyset)$
	\State \Comment {Step 5: Check termination}
	\State \text{GrB\_eWiseMult} $(\mathit{diff}, \emptyset, \emptyset, \text{GxB\_ISNE}, \mathit{dup}, \mathit{gf}, \emptyset)$
	\State \text{GrB\_reduce} $(\&\text{sum}, \emptyset, \text{Add}, \mathit{diff}, \emptyset)$
    \State \text{GrB\_assign} $(\text{dup}, \emptyset, \emptyset, \mathit{gp}, \text{GrB\_ALL}, 0, \emptyset))$
\Until {$\text{sum}=0$}
\EndProcedure
\end{algorithmic}
\end{algorithm}

\subsection{Distributed implementation using CombBLAS.}
\label{sec:sv-combblas}

The distributed version of \Name{} is implemented in CombBLAS~\cite{combblas}.
CombBLAS provides all operations needed for \Name{}, but its API differs from the GraphBLAS standard.
GraphBLAS's \textit{collections} (matrices and vectors) are opaque datatypes whose internal representations (sparse or dense) are not exposed to users, but CombBLAS distinguishes them in the user interface.
Then, GraphBLAS's functions often consist of multiple operations (like masking, accumulation and the main operation) as described in \autoref{sec:sv-graphblas}, while in CombBLAS we usually perform a single operation at a time.
Despite these differences, a straightforward implementation of \Name{} on CombBLAS can be obtained by transforming each GraphBLAS function to the semantically equivalent ones in CombBLAS, using dense vectors in all scenarios.

The parallel complexity of the main linear algebraic operations used in \Name{} (the vector variants of \texttt{GrB\_extract} and \texttt{GrB\_assign}, and the \texttt{GrB\_mxv}), as well as the potential optimizations are discussed in the LACC paper~\cite{lacc}.
Due to the similarity of \Name{} and LACC in the algorithm logic, they can be optimized by the similar optimization techniques.
We briefly summarize them below.

\textbf{Broadcasting-based implementation for the extract and assign operations.}
The \emph{extract} and \emph{assign} operations fetch or write data on the specified locations of a vector, which may cause a load balancing issue when there is too much access on a few locations.
In \Name{}, these locations are exactly the set of parent vertices in the pointer graph, and due to the skewed structure of the pointer graph, the root vertices (especially those belonging to a large component) will have extremely high workload.
When using the default \emph{assign} and \emph{extract} implementations in CombBLAS via all-to-all communication, several processes become the bottleneck and slow down the whole operation significantly.
The solution is a manual implementation of these two operations via the detection of the hot spots and broadcasting the entries on those processes.

\textbf{Taking advantage of the sparsity.}
The matrix-vector multiplication $\mathit{mngf}=\mathbf{A}\cdot\mathit{gf}$ is an expensive operation in \Name{} (see our performance profiling in \autoref{sec:sparsity}).
The straightforward implementation naturally chooses the sparse-matrix dense-vector (SpMV) multiplication, since all the vectors in \Name{} are dense vectors.
Alternatively, we can use an incremental implementation by computing $\Delta\mathit{mngf}=\mathbf{A}\cdot(\Delta\mathit{gf})$, where $\Delta\mathit{gf} = \mathit{gf} - \mathit{gf}_{\mathit{prev}}$ containing only the modified entries of $\mathit{gf}$ is stored as a sparse vector, so the multiplication is the sparse-matrix sparse-vector multiplication (SpMSpV)~\cite{azad2017work}.
Depending on the sparsity of $\Delta\mathit{gf}$, SpMSpV could have much lower computation and communication cost than SpMV.
We use a threshold on the portion of modified entries of $\mathit{gf}$ to decide which method to use in each iteration, which effectively reduces the computation time.
\autoref{sec:sparsity} presents a detailed evaluation.

\begin{tabular}{c|c|c|c}
\end{tabular}

\section{Experiments}
\label{sec:evaluation}

\begin{table*}[t]
\centering
\caption{Graph datasets used to evaluate the parallel connected component algorithms.}
\vspace{-5pt}
\footnotesize
\label{tab:datasets}
\begin{tabular}{lrrrl}
\hline
Graph & Vertices & Directed edges & Components & Description \\
\hline
Queen\_4147 & 4.15M & 166.82M & 1 & 3D structural problem~\cite{davis2011university} \\
kmer\_A2a & 170.73M & 180.29M & 5353 & Protein k-mer graphs from GenBank~\cite{davis2011university} \\
archaea & 1.64M & 204.78M & 59794 & archaea protein-similarity network~\cite{hipmcl} \\
kmer\_V1r & 214.01M & 232.71M & 9 & Protein k-mer graphs, from GenBank~\cite{davis2011university} \\
HV15R & 2.02M & 283.07M & 1 & Computational Fluid Dynamics Problem~\cite{davis2011university} \\
uk-2002 & 18.48M & 298.11M & 1990 & 2002 web crawl of .uk domain~\cite{davis2011university} \\
eukarya & 3.24M & 359.74M & 164156 & eukarya protein-similarity network~\cite{hipmcl} \\
uk-2005 & 39.45M & 936.36M & 7727 & 2005 web crawl of .uk domain~\cite{davis2011university} \\
twitter7 & 41.65M & 1.47B & 1 & twitter follower network~\cite{davis2011university} \\
SubDomain & 82.92M & 1.94B & 246969 & 1st-level subdomain graph extracted from Hyperlink~\cite{meusel2014graph} \\
sk-2005 & 50.64M & 1.95B & 45 & 2005 web crawl of .sk domain~\cite{davis2011university} \\
MOLIERE\_2016 & 30.22M & 3.34B & 4457 & automatic biomedical hypothesis generation system~\cite{davis2011university} \\
Metaclust50 & 282.20M & 37.28B & 15982994 & similarities of proteins in Metaclust50~\cite{hipmcl} \\
Hyperlink & 3.27B & 124.90B & 29360027 & hyperlink graph extract from the Common Crawl~\cite{meusel2014graph}\\
\hline
\end{tabular}
\end{table*}

\begin{figure*}[t]
\centering
\includegraphics[width=1.0\textwidth]{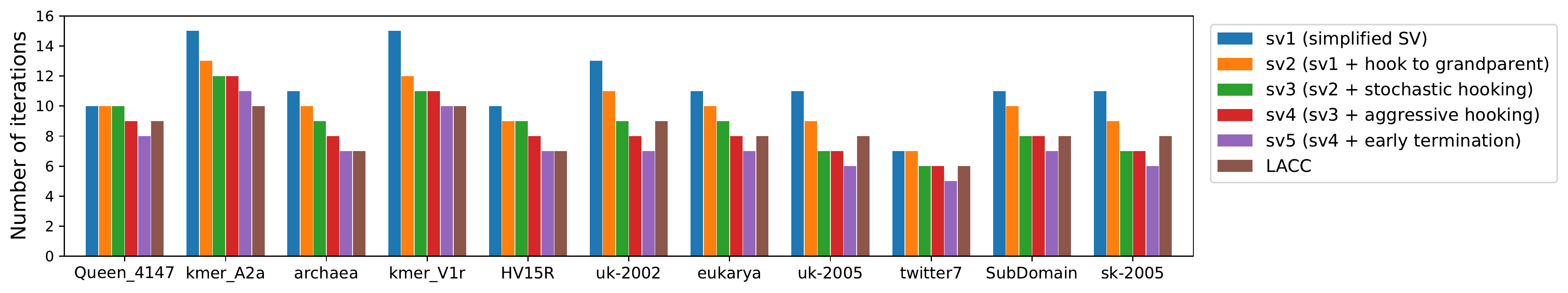}
\vspace{-20pt}
\caption{Number of iterations the simplified SV takes after performing each of the optimizations (sv5 is exactly our \Name{} algorithm), and the number of iterations LACC takes.}
\vspace{-10pt}
\label{fig:sv-iters}
\end{figure*}

\begin{figure}[t]
\centering
\includegraphics[width=0.48\textwidth]{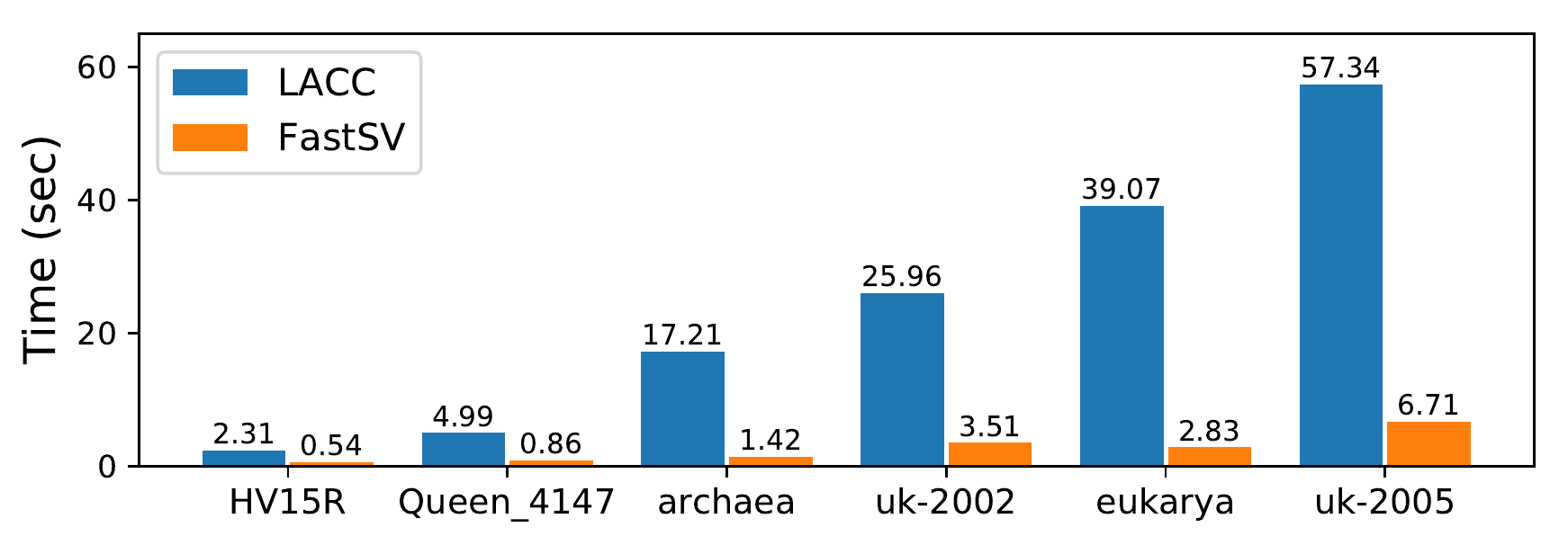}
\vspace{-5pt}
\caption{Performance of the parallel \Name{} and LACC in SuiteSparse:GraphBLAS on six small graphs.}
\vspace{-10pt}
\label{fig:LAGraph}
\end{figure}

\begin{figure*}[t]
\centering
\includegraphics[width=0.95\textwidth]{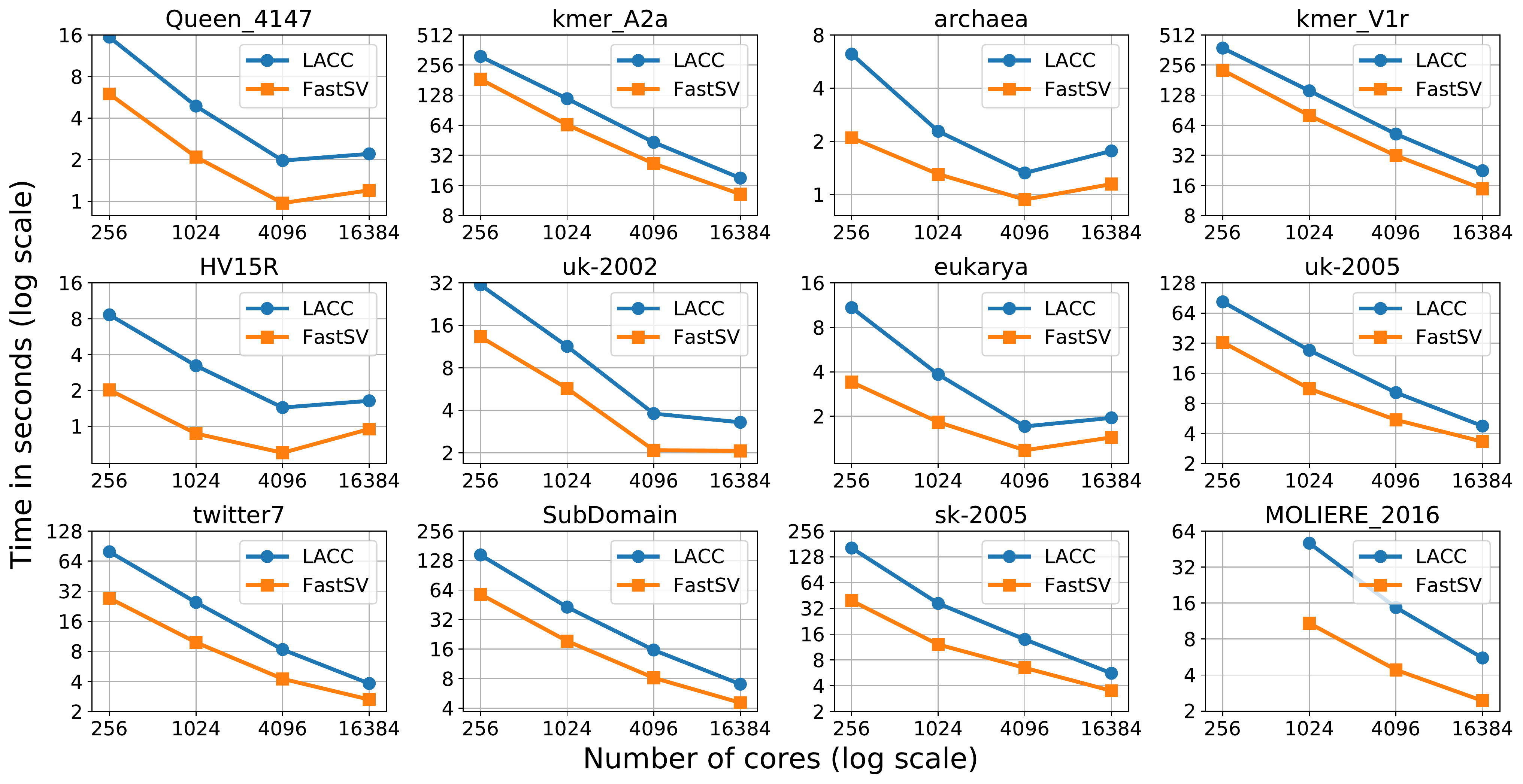}
\vspace{-5pt}
\caption{Strong scaling of distributed-memory \Name{} and LACC using up to $16384$ cores (256 nodes).}
\vspace{-10pt}
\label{fig:scalability}
\end{figure*}

\begin{figure}[t]
\centering
\includegraphics[width=0.48\textwidth]{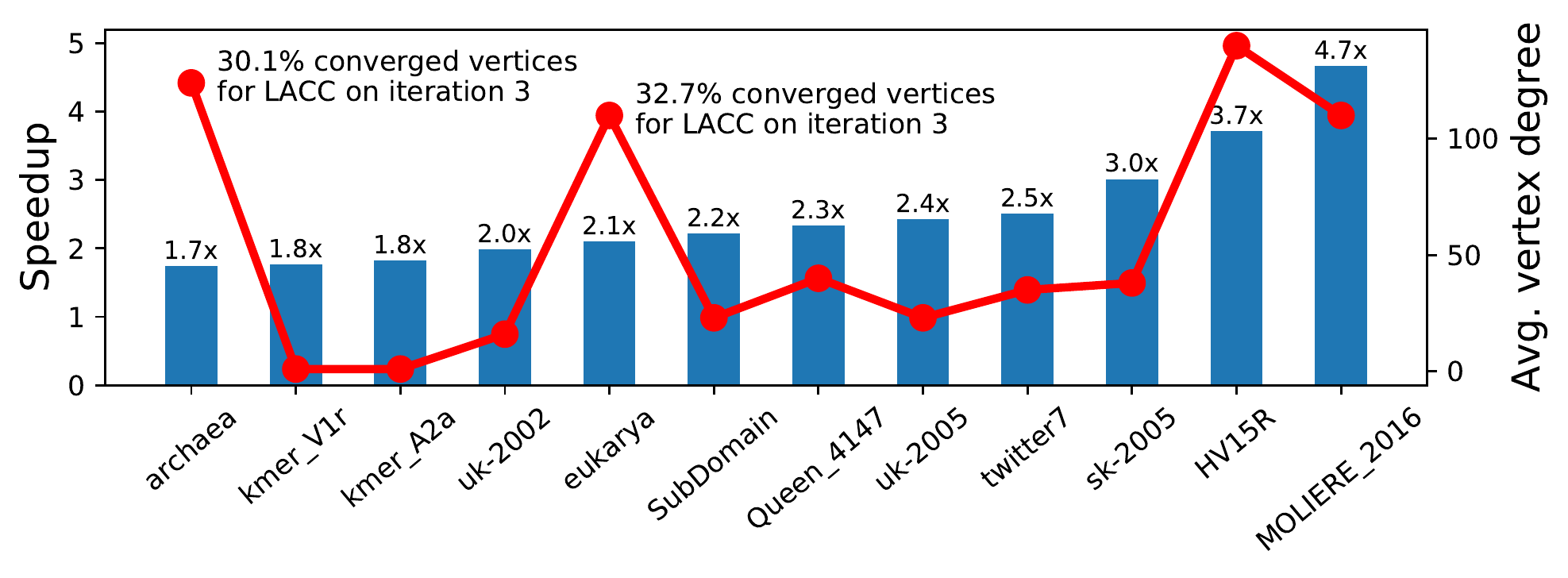}
\vspace{-5pt}
\caption{The speedup of \Name{} over LACC on twelve small datasets using 256 cores (bar chart) and each graph's density in terms of average vertex degree (line chart). A positive correlation between the two metrics can be observed, except for the two outliers \textit{archaea} and \textit{eukarya}.} 
\vspace{-10pt}
\label{fig:density-speedup}
\end{figure}

In this section, we evaluate various aspects of \Name{} showing its fast convergence, shared- and distributed-memory performance, scalability and several other performance characteristics.
We compare \Name{} with LACC~\cite{lacc} that has demonstrated superior performance over other distributed-memory parallel CC algorithms. 
\autoref{tab:datasets} shows a diverse collection of large graphs used to evaluate CC algorithms. 
To the best of out knowledge, the Hyperlink graph~\cite{meusel2014graph} with 3.27B vertices and  124.90B edges is the largest publicly available graph.

\subsection{Evaluation platform.}
We evaluate the performance of distributed algorithms on NERSC Cori supercomputer.
Each node of Cori has Intel KNL processor with 68 cores and 96GB of memory.
All operations in CombBLAS are parallelized with OpenMP and MPI.
Given $p$ MPI processes, we always used a square $\sqrt{p} \times \sqrt{p}$  process grid.  
In our experiments, we used 16 threads per MPI process. 
The execution pattern of our distributed algorithm follows the bulk synchronous parallel (BSP) model, where all MPI processes perform
local computation followed by synchronized communication rounds. 

We also show the shared-memory performance of \Name{} in the SuiteSparse:GraphBLAS library~\cite{suitesparse}.
These experiments are conducted on Amazon EC2's r5.4xlarge instance (128G memory, 16 threads). 


\subsection{Speed of convergence.}
\label{sec:convergence}

At first, we show how different hooking strategies impact the convergence of SV and \Name{} algorithms. 
We start with the simplified SV algorithm (\autoref{algo:algo-one}) and incrementally add different hooking  strategies as shown in \autoref{fig:sv-iters}.
The rightmost bars report the number of iterations needed by LACC. 


\autoref{fig:sv-iters} shows that the simplified SV without unconditional hooking can take up to $1.57\times$ more iterations than LACC. 
We note that despite needing more iterations, \autoref{algo:algo-one} can run faster than LACC in practice because each iteration of the former is faster than each iteration of the latter.
\autoref{fig:sv-iters} demonstrates that SV converges faster as we incrementally apply advanced hooking strategies.
In fact, every hooking strategy improves the convergence of some graphs, and their combination improves the convergence of all graphs. 
Finally, the early termination discussed in \autoref{sec:termination} always removes an additional iteration needed by other algorithms. 
With all improvements, sv5 which represents \autoref{algo:algo-two}, on average reduces $35.0\%$ iterations (min $20\%$, max $46.2\%$) from \autoref{algo:algo-one}.
Therefore,  \Name{} converges as quickly as, or faster than, LACC.

\subsection{Performance in shared-memory platform using SuiteSparse:GraphBLAS.}

To check the correctness of \autoref{algo:graphblas2}, we implemented it in SuiteSparse:GraphBLAS, a multi-threaded implementation of the GraphBLAS standard.
LACC also has an unoptimized SuiteSparse:GraphBLAS implementation available as part of the LAGraph library~\cite{lagraph}.
We compare the performance of \autoref{fig:LAGraph} and LACC in this setting on an Amazon EC2's r5.4xlarge instance with 16 threads.
\autoref{fig:LAGraph} shows that \Name{} is significantly faster than LACC (avg. $8.66\times$, max $13.81\times$).
Although both algorithms are designed for distributed-memory platforms, we still observe better performance of \Name{}, thanks to its simplicity.

\subsection{Performance in distributed-memory platform using CombBLAS.}

We now evaluate the performance of \Name{} implemented using CombBLAS and compare its performance with LACC on the Cori supercomputer.
Both algorithms are implemented in CombBLAS, so they share quite a lot of common operations and optimization techniques (see \autoref{sec:sv-combblas}), making it a fair comparison between the two algorithms.
Generally, \Name{} operates with simpler computation logic and uses less expensive parallel operations than LACC.
However, depending on the structure of the graph, LACC can detect the already converged connected components on the fly and can potentially use more sparse operations. 
Hence, the structure of the input graph often influences the relative performance of these algorithms.

\autoref{fig:scalability} summarizes the performance of \Name{} and LACC on twelve small datasets.
We observe that both \Name{} and LACC  scale to $4096$ cores on all the graphs, and for the majority of the graphs (8 out of 12), they continue scaling to $16384$ cores.
The four graphs on which they stop scaling are just too small that both algorithms finish within 2 seconds.
\Name{} outperforms LACC on all instances.
On 256 cores, \Name{} is $2.80\times$ faster than LACC on average (min $1.66\times$, max $4.27\times$).
When increasing the number of nodes, the performance gap between \Name{} and LACC shrinks slightly, but \Name{} is still $2.53\times$, $1.97\times$ and $1.61\times$ faster than LACC on average on 1024, 4096 and 16384 cores, respectively.

To see how the performance of \Name{} and LACC are affected by the graph structure, we plot the average degree ($|E|/|V|$) and the speedup of \Name{} over LACC for each graph (using 1024 cores) in \autoref{fig:density-speedup}.
Generally, \Name{} tends to outperform LACC by a significant margin on denser graphs.
This is mainly due to fewer matrix-vector multiplications used in \Name{}, whose parallel complexity is highly related to the density of the graph.
The outliers \textit{archaea} and \textit{eukarya} are graphs with a large number of small connected components: they have more than $30\%$ converged vertices detected early. 
On such graphs, LACC's detection of converged connected components provides it with better opportunities to employ
sparse operations, while such detection is not allowed in \Name{}.
Nevertheless, LACC's sparsity optimization still cannot compensate its high computational cost in each iteration.

\subsection{Performance of \Name{} with bigger graphs.}

\begin{figure}[t]
\centering
\includegraphics[width=0.35\textwidth]{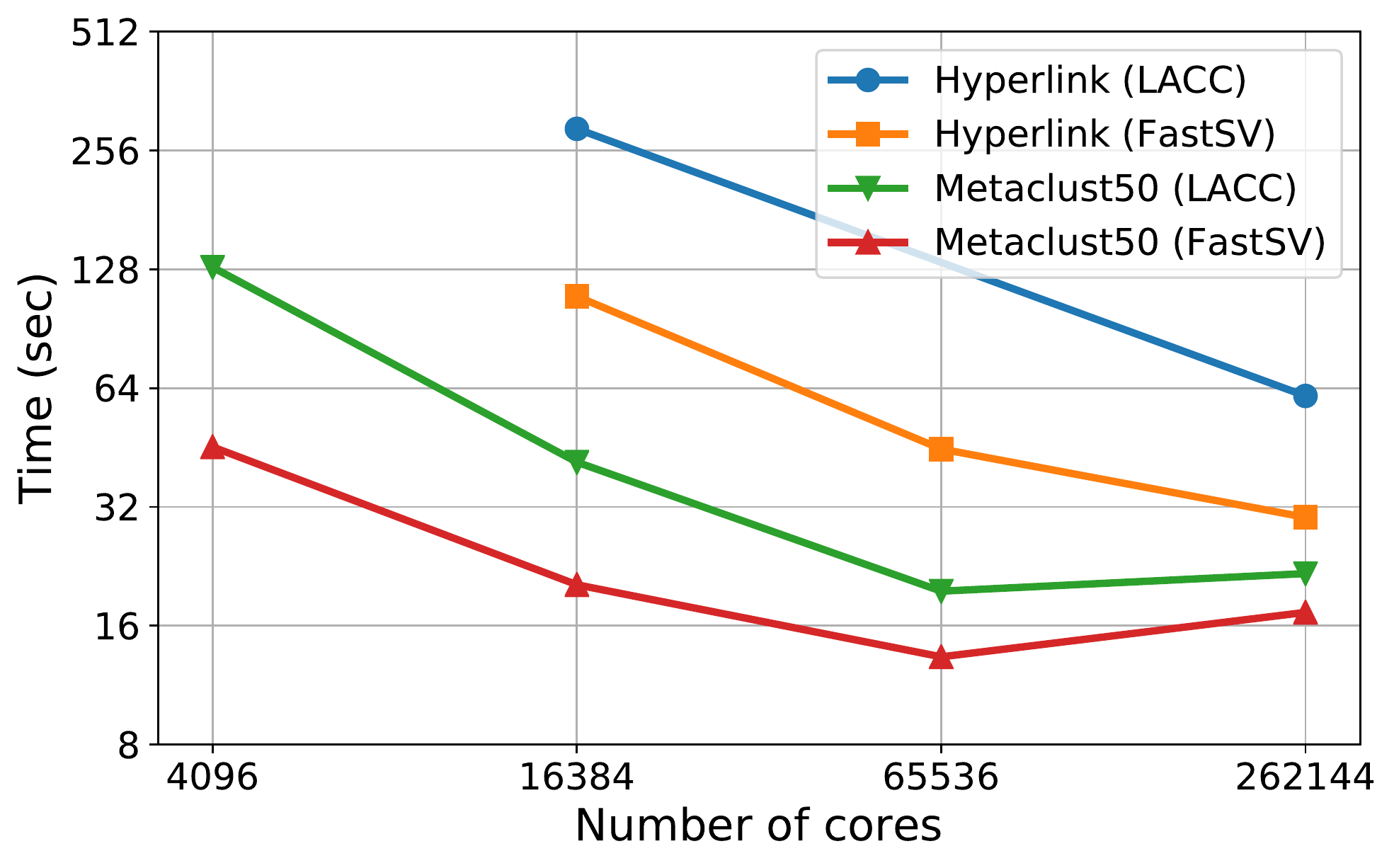}
\vspace{-5pt}
\caption{Performance of \Name{} and LACC with two large graphs on CoriKNL (up to 262144 cores using 4096 nodes).}
\vspace{-10pt}
\label{fig:large}
\end{figure}

We separately analyze the performance of \Name{} and LACC on the two largest graphs in \autoref{tab:datasets}.
Hyperlink is perhaps the largest publicly available graph, making it the largest connectivity problem we can currently solve.
Since each of these two graphs requires more than 1TB memory, it may be impossible to process them on a typical shared-memory server.
\autoref{fig:large} shows the strong scaling of both algorithms and the better performance of \Name{}.
On the smaller graph Metaclust50, both algorithms scale to $65,536$ cores where \Name{} is $1.47\times$ faster than LACC.
On the Hyperlink graph containing 124.9 billion edges, they continue scaling to $262,144$ cores, where \Name{} achieves an $2.03\times$ speedup over the LACC algorithm.



\begin{figure}[t]
\centering
\includegraphics[width=0.44\textwidth]{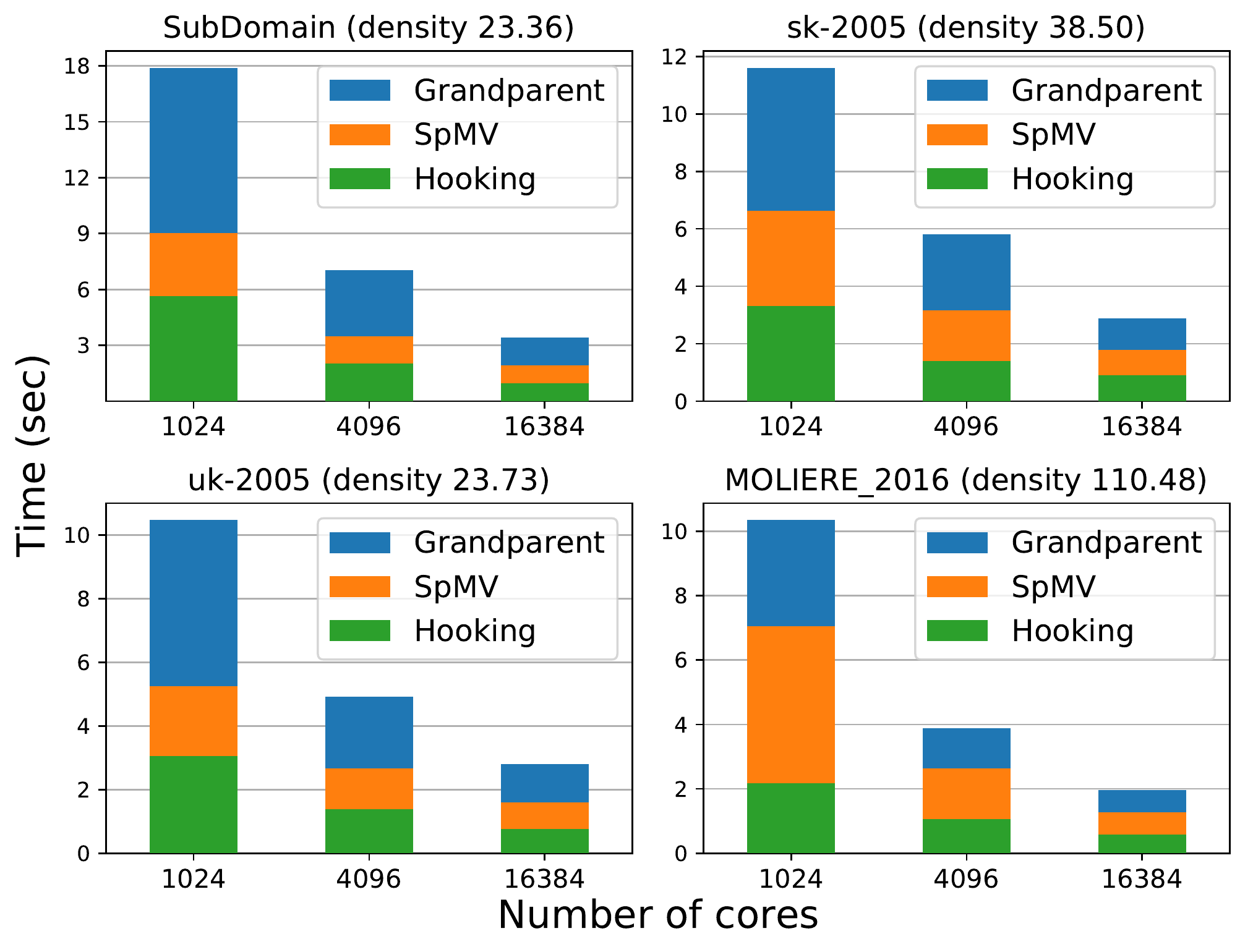}
\vspace{-5pt}
\caption{Performance breakdown of \Name{} on four representative graphs.}
\vspace{-8pt}
\label{fig:sv-parts}
\end{figure}

\subsection{Performance characteristics for operations.}
\label{sec:breakdown}

\autoref{fig:sv-parts} shows the execution time of \Name{} by breaking the runtime into three parts: finding the grandparent, matrix-vector multiplication, the hooking operations.
The time spent on checking the termination is omitted, since it is insignificant relative to other operations.
Each of these operations contributes significantly to the total execution time.
Finding the grandparent and the hooking operations basically reflect the parallel complexity of the \texttt{extract} and \texttt{assign} operations, and the ratio of them is relatively stable for all graphs.
By contrast, the execution time of SpMV varies considerably across different graphs, because SpMV's complexity depends on the density of a graph.

\begin{figure}[t]
\centering
\includegraphics[width=0.48\textwidth]{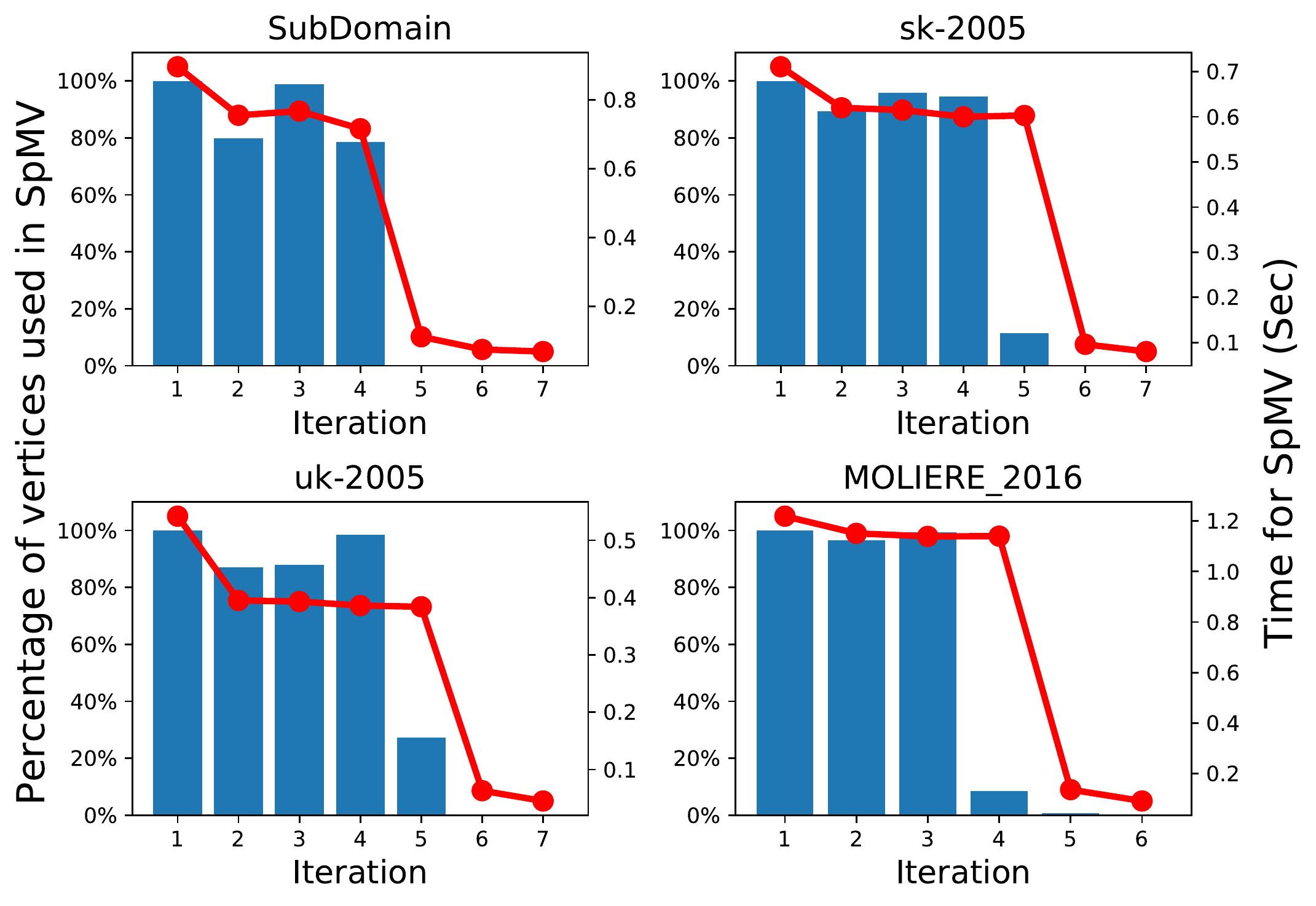}
\vspace{-5pt}
\caption{Percentage of vertices that participate in the SpMV (sparse matrix-vector multiplication) operation for each iteration (bar chart), and the runtime for SpMV (line chart). A vertex participates the SpMV if its grandparent $\mathit{gf}$ is not changed in the previous iteration.}
\vspace{-10pt}
\label{fig:spmv-selected}
\end{figure}

\subsection{Execution time reduced by the sparsity optimization.}
\label{sec:sparsity}

As mentioned in ~\autoref{sec:sv-combblas}, \Name{} dynamically selects SpMV or SpMSpV based on the changes in the grandparent vector $\mathit{gf}$.
This optimization is particularly effective for high-density graphs where SpMV usually dominates the runtime (see  \autoref{fig:sv-parts}).
\autoref{fig:spmv-selected} explains the benefit of sparsity with four representative graphs, where we plot the number of vertices modified in each iteration.
We observe that only a small fraction of vertices participates in the last few iterations where SpMSpV can be used instead of SpMV.
As shown by the red runtime lines in \autoref{fig:spmv-selected}, the use of SpMSpV drastically reduces the runtime of the last few iterations. 

\section{Related work}
\label{sec:related}

Finding the connected components of an undirected graph is a well-studied problem in the PRAM model.
Many of these algorithms such as the Shiloach-Vishkin (SV) algorithm assume the CRCW (concurrent-read and concurrent-write model) model.
The SV algorithm~\cite{sv} takes $O(\log n)$ time on $O(m+n)$ processors.
The Awerbuch-Shiloach (AS) algorithm~\cite{as} is a simplification of SV by using a different termination condition.
Transforming the complete SV or AS to distributed-memory is possible~\cite{svppa,lacc}, but the detection of stagnant trees in SV's \textit{unconditional hooking} step is in fact not suitable for a distributed-memory implementation, which introduces considerable computation and communication cost.
Therefore, we based our distributed-memory \Name{} on a simplified SV algorithm preserving only the essential steps and introduce efficient hooking steps to ensure fast convergence in practice.

There are several distributed-memory connected component algorithms proposed in the literature.
Parallel BFS is a popular method that are implemented and optimized in various systems~\cite{parallel-bfs,graphx,powerlyra,gemini}, but its complexity is bounded by the diameter of the graph, so it is mainly used on small-world networks.
LACC~\cite{lacc} is the state-of-the-art algorithm prior to our work, which guarantees the convergence in $\log(n)$ iterations by transforming the complete AS algorithm into linear algebra operations.
\Name{}'s high performance comes from a much simplified computation logic than LACC.
ParConnect~\cite{parconnect} is another distributed-memory algorithm that adaptively uses parallel BFS and SV and dynamically selects which method to use.
For other software architectures, there are Hash-Min~\cite{hashmin} for MapReduce systems and S-V PPA~\cite{svppa} for vertex-centric message passing systems~\cite{pregel}.
MapReduce algorithms tend to perform poorly in the tightly-couple parallel systems our work targets, compared to the loosely-coupled architectures that are optimized for cloud workloads.
The S-V PPA algorithm, due to the requirement of communicating between non-neighboring vertices, is only supported by several Pregel-like systems~\cite{giraph,pregelplus,pregel-channel}, and these frameworks tend to have limited scalability on multi-core clusters due to the lack of multi-threading support.


\section{Conclusion}
\label{sec:conclusion}

In this paper, we present a new distributed-memory connected component algorithm \Name{} that is scalable to hundreds of thousands processors on modern supercomputers.
\Name{} achieves its efficiency  by first keeping the backbones of the Shiloach-Vishkin algorithm and then employing several novel hooking strategies for fast convergence.
\Name{} attains high performance by employing scalable GraphBLAS operations and optimized communication.
Given the generic nature of our algorithm, it can be easily implemented for any computing platforms such as using GraphBLAST~\cite{yang2019graphblast} for GPUs and can be programmed in most programming languages such as using pygraphblas (https://github.com/michelp/pygraphblas) in Python.

Finding CCs is a fundamental operation in many large-scale applications such as  metagenome assembly and protein clustering.
With the exponential growth of genomic data, these applications will generate graphs with billions of vertices and trillions of edges and will use upcoming exascale computers to solve science problems. 
\Name{} is a step toward such data-driven science as it has the ability to process trillion-edge graphs using millions of cores. 
Overall, \Name{} is generic enough to be used with existing libraries and scalable enough to be integrated with massively-parallel applications.



{\bf Acknoledgement.} Funding for AA was provided by Exascale Computing Project (17-SC-20-SC), a collaborative effort of the U.S. Department of Energy Office of Science and the National Nuclear Security Administration.
This work is also partially supported by the Japan Society for the Promotion of Science (JSPS) Grant-in-Aid for Scientific (S) No. 17H06099.
\newpage
\bibliographystyle{abbrv}
\bibliography{ref}

\appendix

\section{Correctness of the early termination}
\label{sec:thm1}

\cref{thm1} states that, in \Name{} if the grandparent $f[f]$ remains unchanged after an iteration, then the vector $f$ will not be changed afterwards.
The proof makes use of the following lemmas.

\begin{lemma}\label{lemma1}
During the whole algorithm, $f[u]\leq u$ holds for all vertices $u$.
\end{lemma}

\begin{proof}
Initially, $f[u]=u$ for all vertex $u$ and the lemma holds trivially.
The the operation $\xleftarrow{\min}$ ensures that $f$ can only decrease, so the lemma always holds.
\end{proof}

\begin{lemma}\label{lemma3}
After an iteration, if the grandparent $f[f]$ remains unchanged, then every vertex hooks onto its grandparent in the previous operation.
\end{lemma}

\begin{proof}
By contradiction.
Suppose $u$ changes its pointer to some $v$ other than $f[f[u]]$, then since it overrides the shortcutting operation $f_{\mathit{next}}[u]\xleftarrow{\min}f[f[u]]$ we know that $v<f[f[u]]$.
By \cref{lemma1}, $u$'s new grandparent $f_{\mathit{next}}[v]\leq v<f[f[u]]$, then the grandparent of $u$ is changed.
\end{proof}

\begin{lemma}\label{lemma4}
After an iteration, if the grandparent $f[f]$ remains unchanged, then every vertex points to a root now.
\end{lemma}

\begin{proof}
By contradiction.
Suppose $u$'s new parent $v$ is not a root, then $u$'s new grandparent is $f_{\mathit{next}}[v]<v=f[f[u]]$ (by \cref{lemma1} and \cref{lemma3}), which means $u$'s grandparent has changed.
\end{proof}

Here we prove \cref{thm1}.

\begin{proof}
We show that no hooking operation will be performed if $f[f]$ remains unchanged after an iteration.
The aggressive hooking in the form of $f_{\mathit{next}}[u]\xleftarrow{\min}f[f[v]]$ is overridden by the shortcutting operation $f_{\mathit{next}}[u]\xleftarrow{\min}f[f[u]]$ in the previous iteration (by \cref{lemma3}), meaning that $f[f[u]]\leq f[f[v]]$ for all $(u,v)\in E$.
Since then, $f[f]$ is not changed, so the aggressive hooking will not be performed in the current iteration either.
The stochastic hooking $f_{\mathit{next}}[f[u]]\xleftarrow{\min}f[f[v]]$ will not be performed since for all $(u,v)\in E$ we have $f[f[u]]\leq f[f[v]]$.
Shortcutting will not be performed either since every vertex points to a root now (by \cref{lemma4}).
Then, no hooking operation can be performed, and the vector $f$ remains unchanged afterwards.
\end{proof}

\end{document}